\documentclass[10pt,a4paper,reqno]{amsart}
\usepackage[latin1]{inputenc}
\usepackage{amsmath,enumerate,amsthm,graphicx,color,verbatim,mathtools}
\usepackage{amsfonts,amscd}
\usepackage{latexsym}
\usepackage{amssymb}
\numberwithin{equation}{section}

\newtheorem{thm}{Theorem}

\newtheorem{remark}{Remark}

\newtheorem{lemma}[thm]{Lemma}

\newtheorem{prop}[thm]{Proposition}

\newcommand{\Z}{\mathbb{Z}}
\newcommand{\D}{\mathbb{D}}
\newcommand{\R}{\mathbb{R}}

\newcommand{\s}{\text{\rm{s}}}

\newcommand{\ess}{\text{\rm{ess}}}
\newcommand{\ac}{\text{\rm{ac}}}
\newcommand{\lj}{{\mathbf l}}
\DeclareMathOperator{\Var}{Var}
\DeclareMathOperator{\supp}{supp}

\title[Generalized Pr\"ufer variables for Jacobi and CMV matrices]{Generalized Pr\"ufer variables for perturbations of Jacobi and CMV matrices}

\date\today

\author{Milivoje Lukic}
\address{Department of Mathematics, University of Toronto, Bahen Centre, 40 St. George St., Toronto, Ontario, CANADA M5S 2E4 and Department of Mathematics, Rice University, Houston TX 77005, U.S.A.}

\email{mlukic@math.toronto.edu}

\thanks{M.L.\ was partially supported by NSF Grant DMS--1301582. M.L. would also like to thank the Isaac Newton Institute for Mathematical Sciences, Cambridge, for support and hospitality during the programme ``Periodic and Ergodic Spectral Problems" where part of this work was undertaken.}

\author{Darren C. Ong}
\address{Department of Mathematics, University of Oklahoma, Norman, OK 73019-3103, U.S.A.}
\email{darrenong@math.ou.edu}

\thanks{D.O.\ was partially supported by NSF grant DMS--1067988.}

\subjclass[2010]{47B36,42C05,39A70}

\begin{document}
\begin{abstract}
Pr\"ufer variables are a standard tool in spectral theory, developed originally for perturbations of the free Schr\"odinger operator. They were generalized by Kiselev, Remling, and Simon to perturbations of an arbitrary Schr\"odinger operator. We adapt these generalized Prufer variables to the setting of Jacobi and Szeg\H o recursions. We present an application to random $L^2$ perturbations of Jacobi and CMV matrices, and an application to decaying oscillatory perturbations of periodic Jacobi and CMV matrices.
\end{abstract}

\maketitle

\begin{section}{Introduction}
Let $H_0$ be a second-order differential or difference operator and let $V$ be a perturbation. To analyze spectral properties of the perturbed operator $H_0+V$, it is often useful to compare them to  spectral properties of the, usually simpler, unperturbed operator $H_0$. The comparison can be done at the level of eigensolutions, by which we always mean generalized eigensolutions, i.e.\ formal eigensolutions of the differential or difference operator, not necessarily in the Hilbert space. The goal is then to compare eigensolutions $\varphi$ of $H_0$,
\[
H_0 \varphi = E \varphi,
\]
to eigensolutions $u$ of the perturbed operator $H_0+V$,
\[
(H_0+V) u = E u.
\]
One strategy is to define Pr\"ufer variables $R, \theta$ in a way that quantifies this comparison, so that $R,\theta$ obey a first-order differential (or difference) equation. This strategy has been implemented for several classes of operators, starting with the work of Pr\"ufer \cite{Prufer} for Schr\"odinger operators, in the case where $H_0 = - \Delta$ is the free Laplacian and $\varphi(x) = e^{ikx}$. For perturbations of the free Jacobi matrix, the analogous variables arose gradually in the work of several authors, first for discrete Schr\"odinger operators \cite{Eggarter,GredeskulPastur,PasturFigotin}, later also for more general Jacobi matrix perturbations \cite{Breuer07,Breuer10,BreuerLastSimon10,KaluzhnyLast07,Lukic-OP}. Pr\"ufer variables for orthogonal polynomials on the unit circle, for perturbations of the free case, were first used in \cite{Nikishin}; see also \cite{SimonOPUC2}. 

Pr\"ufer variables are a very important tool for analyzing properties of eigensolutions and, since properties of eigensolutions imply spectral properties of $H_0+V$, they have been used extensively in spectral theory, especially in the study of decaying perturbations; see, e.g., \cite{KiselevLastSimon98,SimonOPUC2,Lukic-infinite}.

Furthermore, \cite{KRS} developed \em generalized \em Pr\"ufer variables, which extend this approach to the case where $H_0$ is an arbitrary, continuous or discrete, Schr\"odinger operator (instead of just the free Laplacian). Those variables are well suited to the study of decaying perturbations of $H_0$, especially in cases where eigensolutions of $H_0$ have good properties, e.g. for periodic $H_0$; see, e.g., \cite{KRS,Lukic--Ong}.

In this paper, we adapt the idea of generalized Pr\"ufer variables to two other difference equations, the \em Jacobi \em and \em Szeg\H o \em recursions, which correspond to Jacobi and CMV matrices and, equivalently, to orthogonal polynomials on the real line and orthogonal polynomials on the unit circle. We will explain the setup in the introduction, postponing the details to later sections.

\begin{subsection}{The Jacobi recursion}
We consider a Jacobi matrix $\mathcal J$ with coefficients $a_n > 0$, $b_n \in \mathbb{R}$, 
\begin{equation}
\mathcal J=\left(
\begin{array}{ccccc}
b_1 & a_1 & 0  &0&\cdots \\
a_1 &b_2 &a_2&0&\cdots\\
0  &a_2&b_3&a_3&\cdots \\
0  &  0&a_3&b_4&\cdots\\
\cdots&\cdots&\cdots&\cdots&\cdots
\end{array}
\right).
\end{equation}
The Jacobi matrix is viewed as an operator on $\ell^2(\mathbb N_0)$.

We consider also its perturbation, a Jacobi matrix $\tilde {\mathcal J}$ with coefficients $a_n + a_n' > 0$, $b_n + b_n' \in \mathbb{R}$. Consider, for  $E\in \R$,  
a solution $\varphi$ of the eigenvalue equation $\mathcal J\varphi=E\varphi$, that is,
\begin{equation}\label{varphi}
a_{n+1}\varphi(n+1)+b_{n+1}\varphi(n)+a_n \varphi(n-1) = E\varphi(n),
\end{equation} 
and an eigensolution $u$ for $\tilde {\mathcal J}$,
\begin{equation}\label{u}
(a_{n+1}+ a_{n+1}')u(n+1)+(b_{n+1}+ b_{n+1}')u(n)+(a_n+ a_n') u(n-1) = Eu(n).
\end{equation} 
At this point there will be an assymetry in our setup. We assume that $\varphi$ is linearly independent with its complex conjugate $\bar\varphi$ (we refer to this as a ``complex" solution from now on). On the other hand, we assume that $u$ is a real-valued eigensolution.

We can now define the Pr\"ufer variable $Z(n)$ by
\begin{align}
\label{rho1}
\begin{pmatrix}
(a_n+ a_n')u(n)\\
u(n-1)
\end{pmatrix}=&\frac{1}{2i}\left(
Z(n)\begin{pmatrix}
a_n\varphi(n)\\
\varphi(n-1)
\end{pmatrix}
-\overline {Z(n)}
\begin{pmatrix}
a_n\overline{\varphi(n)}\\
\overline{\varphi(n-1)}
\end{pmatrix}
\right)
\\
=&
\label{rho2}\mathrm{Im}\left[
Z(n)
\begin{pmatrix}
a_n\varphi(n)\\
\varphi(n-1)
\end{pmatrix}
\right].
\end{align}
By linear independence of $\varphi$ and $\bar\varphi$ and reality of $u$, \eqref{rho1} uniquely determines $Z(n)$. We also define the Pr\"ufer amplitude $R(n) > 0$ and Pr\"ufer phase $\eta(n) \in \mathbb{R}$ by
\begin{equation}\label{ZReta}
Z(n) = R(n) e^{i\eta(n)}.
\end{equation}
The second-order linear equation \eqref{u} reduces to a first-order nonlinear recursion relation for $Z(n)$, which we derive in Section~\ref{JacobiPrufer}.

We point out again that this approach was introduced by \cite{KRS} for the case $a_n =1$, $a_n'=0$.

\end{subsection}
\begin{subsection}{The Szeg\H o recursion}
For $z\in \partial\D$ and $\alpha \in \D$, introduce transfer matrices
\begin{equation}\label{Aalphaz}
A(\alpha,z)=\frac{1}{\sqrt{1-\lvert \alpha\rvert^2}}
\begin{pmatrix}
z &-\overline{\alpha}\\
-\alpha z& 1
\end{pmatrix}.
\end{equation}
For a sequence of Verblunsky coefficients $\{\alpha_n\mid n\in \mathbb N_0\}$, consider the Szeg\H o recursion given by
\begin{equation}\label{OPUC.tmatrix}
v(n+1)=
z^{-1/2}A(\alpha_n,z)
v(n)
\end{equation}
where $v(n) \in \mathbb C^2$. Szeg\H o recursion is commonly stated without the $z^{-1/2}$, but the factor is added here out of convenience. The choice of square root will be irrelevant in what follows, as long as it is consistent between formulas.

Szeg\H o recursion generates orthogonal polynomials on the unit circle: if $\alpha_n$ are the Verblunsky coefficients corresponding to a measure $\mu$ on the unit circle and $v(0) = \begin{pmatrix} 1 \\ 1 \end{pmatrix}$, the corresponding orthogonal polynomials $\varphi_n(z)$ obey
\[
v(n) = \begin{pmatrix} z^{-n/2} \varphi_n(z) \\ z^{n/2} \overline{\varphi_n(1/\bar z)} \end{pmatrix}.
\]
Thus, Szeg\H o recursion plays the role for orthogonal polynomials on the unit circle that Jacobi recursion plays for orthogonal polynomials on the real line. The corresponding matrix representation is given by the CMV matrix, 
\begin{equation}
\mathcal C=\left(
\begin{array}{cccccc}
\overline{\alpha_0}&\overline{ \alpha_1}\rho_0&\rho_1\rho_0&0&0&\ldots\\
\rho_0 &- \overline{\alpha_1}\alpha_0&-\rho_1\alpha_0&0&0&\ldots\\
0&\overline{\alpha_2}\rho_1&-\overline{\alpha_2}\alpha_1&\overline{\alpha_3}\rho_2&\rho_3\rho_2&\ldots\\
0&\rho_2\rho_1&-\rho_2\alpha_1&-\overline{\alpha_3}\alpha_2&-\rho_3\alpha_2&\ldots\\
0&0&0&\overline{\alpha_4}\rho_3&-\overline{\alpha_4}\alpha_3&\ldots\\
\ldots& \ldots&\ldots&\ldots&\ldots&\ldots\\
\end{array}
\right),
\end{equation}
a unitary operator from $\ell^2(\mathbb N)$ to $\ell^2(\mathbb N)$; here we denote $\rho_n = \sqrt{1-\vert \alpha_n\vert^2}$.

The CMV operator is of central importance in the theory of orthogonal polynomials on the unit circle. In particular, it can be understood as the unitary analogue of the self-adjoint Jacobi operator. For more details on the Szeg\H o recursion and the CMV operator, please refer to \cite{SimonOPUC1,SimonOPUC2,Simon2}.

Consider now a perturbation $\alpha' = \{ \alpha_n'\}_{n=0}^\infty$ such that $\alpha_n+\alpha'_n\in \D$ and a solution $u$ of the perturbed Szeg\H o recursion
\begin{equation}\label{OPUC.perturbed}
u(n+1)=z^{-1/2}A(\alpha_n + \alpha'_n,z) u(n)
\end{equation}
with an initial condition of the form
\begin{equation}\label{u0}
u(0)=\begin{pmatrix}
\kappa \\
\bar \kappa
\end{pmatrix}.
\end{equation}

Let us define an antilinear operator $C$ by
\begin{equation}
C\begin{pmatrix}
w_1\\
w_2
\end{pmatrix}
=
\begin{pmatrix}
\overline {w_2}\\
\overline {w_1}
\end{pmatrix}
\end{equation}
and let us write $v^*(n)= Cv(n)$. We will show that
\begin{prop}
There is a unique $Z(n) \in \mathbb{C}$ such that
\begin{equation}\label{uvequation}
u(n)=Z(n) v(n)+\overline{Z(n)} v^*(n).
\end{equation}
\end{prop}
The quantity $Z(n)$ is taken as our Pr\"ufer variable in this setting; Pr\"ufer amplitude and Pr\"ufer phase are then defined by \eqref{ZReta}, analogously to the Jacobi case.

Our setup here is different from that for Jacobi matrices: we consider a first-order recursion given by the $2\times 2$ matrices \eqref{Aalphaz} instead of a second-order eigenvector equation. The distinction is not trivial since, for CMV matrices, generalized eigenfunctions are not generated by the Szeg\H o recursion, but rather by the Gesztesy--Zinchenko~\cite{GZ} recursion. However, spectral properties can be characterized directly in terms of the transfer matrices associated with Szeg\H o recursion, which justifies our setup. Moreover, \cite{DFLY} have noted a simple relation between the two recursions which will also allow us to link our Pr\"ufer variables to asymptotics of eigenfunctions.

The details of the preceding discussion and the first-order recursion relation obeyed by the Pr\"ufer variable are given in Section~\ref{PruferSzego}.
\end{subsection}

\begin{subsection}{Decaying random perturbations} The effect of random decaying perturbations on the spectrum of a Schr\"odinger operator has been studied in several papers \cite{Simon82,DelyonSimonSouillard85,Delyon85,KotaniUshiroya88,BreuerLast07,KaluzhnyLast07,KaluzhnyLast11}. In particular, for random $L^2$ discrete Schr\"odinger operators, Kiselev--Last--Simon \cite{KiselevLastSimon98} presented a simple proof that spectrum is almost surely purely absolutely continuous on $(-2,2)$.

Using the generalized Pr\"ufer variables just introduced, we can extend their argument to random decaying perturbations with a sufficiently nice background matrix.

To state the result, we define transfer matrices as
\[
T_N^E = \prod_{n=N}^1 \begin{pmatrix} \frac{E-b_{n+1}}{a_n}  &  - a_n \\ \frac 1{a_n} & 0 \end{pmatrix}
\]
since this is the definition that best matches our placement of $a_n$ in this work (see, e.g.,\eqref{rho2}). Several other conventions exist in the literature, differing in the placement of $a_n$. For instance, the choice made in \cite{DamanikKillipSimon10} corresponds to
\[
\begin{pmatrix} a_{N+1}^{-1} & 0 \\ 0 & a_N \end{pmatrix} T_N^E \begin{pmatrix} a_{1} & 0 \\ 0 & a_0^{-1} \end{pmatrix}.
\]
Since our result concerns bounded Jacobi matrices, this convention does not make a difference in what follows (e.g. in \eqref{unifbddsols}).

\begin{thm}\label{TrandomJacobi}
Let $\mathcal{J}$ be a bounded Jacobi matrix and $(u,v)$ an interval such that for every $\epsilon > 0$,
\begin{equation}\label{unifbddsols}
\sup_{E \in (u+\epsilon, v- \epsilon)} \sup_{N\in \mathbb{N}} \lVert T_N^E \rVert < \infty.
\end{equation}
Let $a'_n$, $b'_n$ be sequences of real-valued independent random variables such that the following hold:
\begin{equation}
\mathbb{E} (a'_n) = \mathbb{E}(b'_n) = 0,
\end{equation}
\begin{equation}\label{L2randomcondition}
\sum_n \mathbb{E} ({a'_n}^2 + {b'_n}^2) < \infty,
\end{equation}
and almost surely, for all $n$, $a_n + a_n' > 0$. Then, almost surely, $\tilde{\mathcal{J}}$ has purely absolutely continuous spectrum on $(u,v)$, and $(u,v)$ is in the essential support of the absolutely continuous spectrum.
\end{thm}

\begin{remark}
\eqref{unifbddsols} implies that $\mathcal{J}$ has purely a.c.\ spectrum on $(u,v)$ so, by Dombrowski~\cite{Dombrowski78}, $\inf_n a_n > 0$. Then, the conditions of Theorem \ref{TrandomJacobi} imply that $\tilde J$ is almost surely a bounded Jacobi matrix with $\inf_n (a_n + a_n') > 0$, by $\mathbb{P} (\lvert a_n' \rvert \ge M) \le \mathbb{E}({a_n'}^2)) / M^2$ and the Borel--Cantelli lemma.
\end{remark}

The conditions of this theorem are known to hold in a variety of cases, for instance, if the unperturbed Jacobi matrix $\mathcal{J}$ is periodic. More generally, by results of \cite{PeherstorferSteinbauer00}, they also hold if $\mathcal{J}$ obeys a $q$-bounded variation condition for some $q$, i.e.~if
\begin{equation}\label{qvariation}
\sum_{n} \lvert a_{n+q} - a_n \rvert + \lvert b_{n+q} - b_n \rvert < \infty.
\end{equation}
In those cases, it is known that the essential spectrum of $\mathcal{J}$ consists of finitely many bands, and that \eqref{unifbddsols} holds on each band.

In this way, our result generalizes a result of Kaluzhny--Last \cite{KaluzhnyLast07}, who prove Theorem~\ref{TrandomJacobi} in the case that \eqref{qvariation} holds for $q=1$. Their argument is more elaborate, working directly with transfer matrices and using the bounded variation condition. Our argument applies also in other cases where the conditions of Theorem \ref{TrandomJacobi} are known to hold, e.g.\ for the oscillatory decaying Jacobi matrices studied in \cite{Lukic-OP}.

We prove also a version for the CMV operator. For this setting, let us for the sake of notational convenience define
\[ T(N,z)=\prod_{n=N}^1 A(\alpha_n,z).\]

\begin{thm}\label{TrandomCMV}
Let $\mathcal{C}$ be a CMV matrix corresponding to a sequence of Verblunsky coefficients $\alpha$, and $(u,v)$ an arc on the unit circle (parametrized as $[0,2\pi)$ such that for every $\epsilon > 0$,
\begin{equation}\label{unifbddsolsCMV}
\sup_{E \in (u+\epsilon, v- \epsilon)} \sup_{N\in \mathbb{N}} \lVert T(N,z) \rVert < \infty.
\end{equation}
Let $\alpha'_n$ be a sequence of $\mathbb D$-valued independent random variables such that the following hold:
\begin{equation}
\mathbb{E} (\alpha'_n) = 0,
\end{equation}
\begin{equation}\label{L2randomconditionCMV}
\sum_n \mathbb{E} ( {\lvert \alpha'_n\rvert}^2) < \infty,
\end{equation}
and such that for all $n$,
\begin{equation}
 \lvert \alpha_n+\alpha_n' \rvert < 1.
\end{equation}
Then, almost surely, the CMV operator $\tilde{\mathcal{C}}$ corresponding to Verblunsky coefficients $\alpha_n+\alpha'_n$ has purely absolutely continuous spectrum on $(u,v)$, and $(u,v)$ is in the essential support of the absolutely continuous spectrum.
\end{thm}

Theorems ~\ref{TrandomJacobi} and ~\ref{TrandomCMV} are proved in Section \ref{SrandomJacobi}.

\end{subsection}

\begin{subsection}{Decaying oscillatory perturbations}
As another application of our generalized Pr\"ufer variables for Jacobi and Szeg\H o recursions, we prove a result about decaying oscillatory perturbations of periodic Jacobi or CMV operators. We study a class of perturbations which includes finite or infinite linear combinations of Wigner--von Neumann type perturbations $\sin(n\beta) / n^\gamma$, $\gamma >0$. Perturbations of this form have been discussed since a paper of Wigner--von Neumann~\cite{WvN} showed that such a perturbation of the free Laplacian can create a single embedded eigenvalue in the essential spectrum of the Schr\"odinger operator.

Their spectral properties have been well understood in the $L^2$ regime \cite{Atkinson54,HarrisLutz75,KiselevLastSimon98,Wong,Kurasov-Naboko,Naboko-Simonov,Kurasov-Simonov} and, more recently, in the $L^p$ regime with arbitrary $p<\infty$ \cite{JanasSimonov10,Lukic-OP,Lukic-continuous,Lukic-infinite,Lukic--Ong}. Here we will show the analogue of a result from \cite{Lukic--Ong}, proving in great generality that such perturbations preserve absolutely continuous spectrum of a periodic Jacobi or CMV matrix.

The proof is an adaptation of the proof in \cite{Lukic--Ong}, but it also contains a new ingredient. A part of the proof in \cite{Lukic--Ong} relies on cancellations established by some explicit calculations, which would have been prohibitively long in the current setting. Instead, the calculations are replaced by an indirect argument, which uses general observations to show that the cancellations must occur;  these are stated in Lemmas~\ref{InitialConditionsTrick} and \ref{InitialConditionsTrickCMV} below.

\begin{thm} \label{maintheorem} Consider a $q$-periodic Jacobi matrix $\mathcal J$ with coefficients $a_n, b_n$ and its perturbation  $\tilde {\mathcal J}$ with Jacobi coefficients $a_n + a_n'$,  $b_n + b_n'$. Assume that the sequences $\{a_n'\}_{n=1}^\infty$, $\{b_n'\}_{n=1}^\infty$ can be written in the form
\begin{equation}\label{1.13}
a_n' =  \sum_{l=1}^\infty c_{2l-1} e^{-in\phi_{2l-1}} \varsigma^{(2l-1)}_{n}, \qquad b_{n+1}' =  \sum_{l=1}^\infty c_{2l} e^{-in\phi_{2l}} \varsigma^{(2l)}_{n},
\end{equation}
where $c_l\in\mathbb{C}$, $\phi_l\in\mathbb{R}$, and for some integer $p\ge 2$ and some real number $\beta \in (0,\frac 1{p-1})$, the following conditions hold:
\begin{enumerate}[(i)]
\item (uniformly bounded variation)
\begin{equation}\label{tau}
\tau=\sup_l \sum_{n=1}^\infty \lvert \varsigma_{n+1}^{(l)} - \varsigma_n^{(l)} \rvert <  \infty;
\end{equation}
\item (uniform $\ell^p$ condition)
\begin{equation}
\sup_l \sum_{n=1}^\infty \lvert \varsigma_n^{(l)} \rvert^p < \infty;
\end{equation}
\item (decay of coefficients)
\begin{equation}\label{decayofcoefficients}
\sum_{l=1}^\infty \lvert c_l \rvert^{\beta} < \infty.
\end{equation}

\end{enumerate}
Then there is a set $S \subset \sigma_\ess(\mathcal J)$ of Hausdorff dimension at most $\beta(p-1)$ such that for $E \in \sigma_\ess(\mathcal J) \setminus S$, all eigensolutions of  $\tilde{\mathcal J}$ are bounded. In particular, 
\[
\sigma_\ac(\tilde {\mathcal J}) = \sigma_\ac(\mathcal J)
\]
and
\[
\tilde\mu_\s( \sigma_\ess(\mathcal J) \setminus S)=0,
\]
where $\tilde \mu_\s$ denotes the singular part of the canonical spectral measure corresponding to $\tilde {\mathcal J}$.
\end{thm}

\begin{remark}
In \eqref{1.13}, $b'$ is indexed with $n+1$ purely for later notational convenience.
\end{remark}

\begin{thm} \label{maintheoremCMV}Consider a $q$-periodic CMV matrix $\mathcal C$ with Verblunsky coefficients $\alpha(n)$ and its perturbation  $\tilde {\mathcal C}$ with Verblunsky coefficients $\alpha(n)+\alpha'(n)$. Assume that the sequence $\{\alpha'(n)\}_{n=0}^\infty$ can be written in the form
\begin{equation}\label{alpha'CMV}
\alpha'(n) =  \sum_{l=1}^\infty c_{l} e^{-in\phi_{l}} \varsigma^{(l)}_{n},
\end{equation}
where $c_{l}\in\mathbb{C}$, $\phi_{l}\in\mathbb{R}$, and for some integer $p\ge 2$ and some real number $\beta \in (0,\frac 1{p-1})$, conditions (i), (ii), (iii) of the previous theorem hold.
Then there is a set $S \subset \sigma_\ess(\mathcal C)$ of Hausdorff dimension at most $\beta(p-1)$ such that for $z \in \sigma_\ess(\mathcal C) \setminus S$, all eigensolutions of  $\tilde {\mathcal C}$ are bounded, 
\begin{equation}\label{acC}
\sigma_\ac(\tilde {\mathcal C}) = \sigma_\ac(\mathcal C),
\end{equation}
and
\[
\tilde\mu_\s( \sigma_\ess(\mathcal C) \setminus S)=0,
\]
where $\tilde \mu_\s$ denotes the singular part of the canonical spectral measure of $\tilde {\mathcal C}$.
\end{thm}

Proofs of Theorems~\ref{maintheorem} and \ref{maintheoremCMV} are presented in Section~\ref{SectionDecaying}.

\end{subsection}

\end{section}
\begin{section}{Pr\"ufer variables for the Jacobi recursion}\label{JacobiPrufer}
Let us define $\gamma(n)$ by 
 \begin{equation}
 \varphi(n)=\vert\varphi(n)\vert e^{i\gamma(n)}.
 \end{equation}
We can ensure uniqueness of $\gamma$ by setting $\gamma(0)\in [0,2\pi)$, $\gamma(n)-\gamma(n-1)\in [0,2\pi)$.

Let us also define some variations on the Wronskian. For two sequences $f,g$, we have
\begin{align*}
W_{0,0}(f,g)=& a_{n+1} f(n)g(n+1)-a_{n+1}f(n+1)g(n),\\
W_{a',a'}(f,g)=& (a_{n+1}+a_{n+1}') f(n)g(n+1)-(a_{n+1}+a_{n+1}')f(n+1)g(n),\\
W_{0,a'}(f,g)=& (a_{n+1}+a_{n+1}') f(n)g(n+1)-a_{n+1}f(n+1)g(n).
\end{align*}
The motivation behind the notation $W_{*_1,*_2}(f,g)$ is that we will choose $f$ to be a solution to the Jacobi recursion with the $a_n$ perturbed by $*_1$, and $g$ to be a solution to the Jacobi recursion with the $a_n$ perturbed by $*_2$.

If we assume 
    \[a_{n+1}f(n+1)+a_nf(n-1)=(x-b_{n+1})f(n),\] and \[(a_{n+1}+a_{n+1}')g(n+1)+(a_{n}+a_{n}')g(n-1)=(x-b_{n+1}-b_{n+1}')g(n),\]
   then
 \begin{align}\nonumber
 W_{0,a'}(f,g)(n)-W_{0,a'}(f,g)(n-1)=&-b'_{n+1}f(n)g(n)\\
 &-a'_n(f(n)g(n-1)+f(n-1)g(n)).\label{WronskianDifference}
 \end{align}
 Since $\varphi, \overline\varphi$ are linearly independent solutions of (\ref{varphi}), by constancy of the Wronskian, we have 
 \begin{equation}\label{omegadefn}
 W_{0,0}(\overline\varphi,\varphi)(n)=2ia_{n+1}\mathrm{Im}(\overline{\varphi(n)}\varphi(n+1))=i\omega,
 \end{equation}
 for some real nonzero constant $\omega$. 
Thus,
 \begin{equation}\label{gamma-omega}
 2 \vert\varphi(n)\vert\cdot\vert \varphi(n+1)\vert a_{n+1}\sin(\gamma(n+1)-\gamma(n))=\omega.
\end{equation}
We can use Wronskians to invert (\ref{rho1}) to get
\begin{equation}\label{rho-omega}
Z(n)=\frac{2}{\omega}W_{0,a'}(\overline \varphi,u)(n-1),
\end{equation}
which is the same as (41) in \cite{KRS}.

\begin{thm}\label{t.OPRL} Pr\"ufer variables obey the first-order recursion relation
\begin{align*} \frac{Z(n+1)}{Z(n)}  = &  1-\frac{i}{\omega} \frac{a_n}{a_n+a_n'} b_{n+1}'\vert \varphi(n)\vert^2 (e^{-2i(\eta(n)+\gamma(n))} -1)\\
&+\frac{i}{\omega}  a_{n}'\vert \varphi(n-1)\vert\cdot \vert \varphi(n)\vert e^{i(\gamma(n-1)-\gamma(n))}\\
&-\frac{i}{\omega}  a_{n}'\vert \varphi(n-1)\vert\cdot \vert \varphi(n)\vert e^{-2i\eta(n)}e^{-i(\gamma(n-1)+\gamma(n))}\\
&+\frac{i}{\omega}\frac{a_n}{a_n+a_n'} a_{n}' (1-e^{-2i(\eta(n)+\gamma(n))})\vert \varphi(n-1)\vert\cdot \vert \varphi(n)\vert e^{-i(\gamma(n-1)-\gamma(n))}.
\end{align*}
\end{thm}

This becomes a first-order nonlinear equation for $Z(n)$ if we make the substitution $e^{-2i\eta(n)} = \bar Z(n) / Z(n)$; alternatively, it gives a system of first-order equations for $R(n)$, $\eta(n)$ if we take its absolute value to get a formula for $R(n+1)/R(n)$, and divide it by its complex conjugate to get a formula for $e^{2i(\eta(n+1)-\eta(n))}$.

\begin{proof}[Proof of Theorem~\ref{t.OPRL}]
For notational convenience, let us define
\[\theta(n)=\eta(n)+\gamma(n).\]
By \eqref{rho1},
\begin{align*}
u(n)=& \frac{a_n}{a_n+a_n'}R(n)\vert \varphi(n)\vert
 \sin(\theta(n)),\\
 u(n-1)=& R(n)\vert \varphi(n-1)\vert\sin(\theta(n)-\gamma(n)+\gamma(n-1)).
\end{align*}

We have by (\ref{rho-omega}), (\ref{WronskianDifference}), (\ref{u}), (\ref{varphi}), 
\begin{align}
Z(n+1)-Z(n)=&-\frac{2}{\omega}(b_{n+1}'u(n)\overline{\varphi(n)}+a_n' \overline{\varphi(n)}u(n-1)+a_{n}'\overline{\varphi(n-1)}u(n))\nonumber\\
=&-\frac{2}{\omega}\frac{a_n}{a_n+a_n'}b_{n+1}' R(n)\sin(\theta(n))\vert \varphi(n)\vert^2 e^{-i\gamma(n)}\nonumber\\
&-\frac{2}{\omega}a_{n}' R(n)\sin(\theta(n)-\gamma(n)+\gamma(n-1))\vert \varphi(n)\vert\cdot\vert \varphi(n-1)\vert e^{-i\gamma(n)}\nonumber\\
&-\frac{2}{\omega}\frac{a_n}{a_n+a_n'}a_{n}' R(n)\sin(\theta(n))\lvert \varphi(n-1)\rvert\cdot \lvert \varphi(n)\rvert e^{-i\gamma(n-1)}.
\end{align}
Dividing by $Z(n) = R(n) e^{i\eta(n)}$, we then have
\begin{align*} \frac{Z(n+1)}{Z(n)} - 1 =& -\frac{2}{\omega} \frac{a_n}{a_n+a_n'}b_{n+1}'\vert \varphi(n)\vert^2 \sin (\theta(n)) e^{-i\theta(n)}\\
&- \frac{2}{\omega} a_{n}'\vert \varphi(n-1)\vert\cdot \vert \varphi(n)\vert \sin (\theta(n)+\gamma(n-1)-\gamma(n)) e^{-i\theta(n)}\\
&-\frac{2}{\omega}\frac{a_n}{a_n+a_n'} a_{n}' \sin(\theta(n))\vert \varphi(n-1)\vert\cdot \vert \varphi(n)\vert e^{-i(\theta(n)+\gamma(n-1)-\gamma(n))}.
\end{align*}
Using $e^{-i\theta_n} \sin \theta_n = \frac 1{2i} (1 - e^{-2i\theta_n})$ in the first and third lines and a similar identity in the second line, we obtain

\begin{align*} \frac{Z(n+1)}{Z(n)} =& 1 + \frac{i}{\omega} \frac{a_n}{a_n+a_n'}b_{n+1}'\vert \varphi(n)\vert^2 (1 -  e^{-2i\theta(n)}) \\
& + \frac{i}{\omega} a_{n}'\vert \varphi(n-1)\vert\cdot \vert \varphi(n)\vert e^{i(\gamma(n-1)-\gamma(n))} (1 -  e^{-2i(\theta(n)+\gamma(n-1)-\gamma(n))}) \\
&+ \frac{i}{\omega}\frac{a_n}{a_n+a_n'} a_{n}' \vert \varphi(n-1)\vert\cdot \vert \varphi(n)\vert e^{-i(\gamma(n-1)-\gamma(n))} (1 -  e^{-2i\theta(n)}).
\end{align*}
which completes the proof. \end{proof}

As a first application of the Pr\"ufer variables just introduced, we prove the following theorem.

\begin{thm}\label{l1perturb}
Assume that $\inf_n a_n > 0$ and that $a', b' \in \ell^1$. For any $x\in \mathbb{R}$ such that all solutions $\varphi$ of the recursion relation \eqref{varphi} are bounded, all solutions $u$ of the recursion relation \eqref{u} are bounded as well.
\end{thm}

\begin{proof}
Under the above assumptions, $\lim_{n\to\infty} \frac{a_n}{a_n + a_n'} = 1$ so
\[
M = \sup_n \frac{a_n}{a_n + a_n'} < \infty.
\]
It follows from Theorem~\ref{t.OPRL} that
\begin{equation}\label{ml02}
\left\lvert \frac{Z(n+1)}{Z(n)} - 1 \right\rvert \le \frac {\lVert \varphi\rVert_\infty^2}{\lvert \omega \rvert} \left( 2 M \lvert b'_{n+1} \rvert + 2 \lvert a_n' \rvert + 2 M \lvert a_n'\rvert \right).
\end{equation}
Since the right-hand side of \eqref{ml02} is $\ell^1$, it follows that $Z(n)$ converges as $n\to \infty$ and, in particular, $R(n)$ is bounded in $n$. This implies that $u(n)$ is bounded in $n$ and completes the proof.
\end{proof}

The following lemma will be necessary in the applications which follow.

\begin{lemma}\label{InitialConditionsTrick}
Fix $E$ and a solution $\varphi$ of \eqref{varphi}. For different solutions $u$ of \eqref{u}, denote the corresponding Pr\"ufer variables by $Z_u$, $R_u$, $\eta_u$. If there exists a sequence $\mathcal A(n)$ independent of $u$ such that the series
\begin{equation}\label{ml01}
\sum_{n=1}^\infty \left( \log \frac{Z_u(n+1)}{Z_u(n)}  - \mathcal A(n) \right)
\end{equation}
converges uniformly in solution $u$, then $R_u(n)$ converges as $n\to\infty$ for any solution $u$ and there is no subordinate solution at $E$.
\end{lemma}

\begin{proof}
Let us consider two solutions $u_1(n)$, $u_2(n)$ of \eqref{u}. Subtracting \eqref{ml01} for the two solutions, we conclude that the series
\begin{align*}
\sum_{n=1}^\infty \left( \log \frac{Z_{u_1}(n+1)}{Z_{u_1}(n)}-\log \frac{Z_{u_2}(n+1)}{Z_{u_2}(n)} \right)
\end{align*}
is convergent. In particular, taking real and imaginary parts, we see that the sequences
\[
\log \frac{R_{u_1}(n)}{R_{u_2}(n)}, \qquad  \eta_{u_1}(n) - \eta_{u_2}(n)
\]
converge as $n\to \infty$.

By uniform convergence, there is an $n_0$ such that for all solutions $u$,
\[
\left\lvert \sum_{n=n_0+1}^\infty  \left( \log \frac{Z_u(n+1)}{Z_u(n)}  - A(n) \right)  \right\rvert \le \frac \pi 8.
\]
Taking imaginary parts and subtracting this for $u_1, u_2$,
\[
\left\lvert \sum_{n=n_0+1}^\infty \left( (\eta_{u_1}(n+1) - \eta_{u_1}(n)) - (\eta_{u_2}(n+1) - \eta_{u_2}(n)) \right) \right\rvert \le \frac \pi 4.
\]
Thus,
\[
\left\lvert \lim_{n\to \infty} (\eta_{u_1}(n) - \eta_{u_2}(n)) - ( \eta_{u_1}(n_0) - \eta_{u_2}(n_0) ) \right\rvert \le \frac \pi 4.
\]
In particular, if we were to pick the solution $u_1$ arbitrarily and pick the solution $u_2$ so that
\[
Z_{u_2}(n_0) = i Z_{u_1}(n_0),
\]
then we would have $\eta_{u_2}(n_0) - \eta_{u_1}(n_0) \in \frac \pi 2 + 2\pi \mathbb{Z}$ so
\[
\lim_{n\to\infty} (\eta_{u_2}(n) - \eta_{u_1}(n))  \in \left( \frac \pi 4, \frac {3\pi}4 \right) + 2\pi \mathbb{Z}.
\]

 We consider now the Wronskian of $u_1$ and $u_2$.
 \begin{align*}
& W_{a',a'}(u_1,u_2)(n-1)\\
 =&(a_n+a_n')[u_1(n-1)u_2(n)-u_1(n)u_2(n-1)]\\
 =&a_n\vert \varphi(n)\varphi(n-1)\vert R_{u_2}(n)R_{u_1}(n)[\sin(\eta_{u_2}(n)+\gamma(n))\sin(\eta_{u_1}(n)+\gamma(n-1))\\
 &-\sin(\eta_{u_1}(n)+\gamma(n))\sin(\eta_{u_2}(n)+\gamma(n-1))]\\
  =&a_n\vert \varphi(n)\varphi(n-1)\vert R_{u_2}(n)R_{u_1}(n)\sin (\eta_{u_1}(n)-\eta_{u_2}(n))\sin(\gamma(n)-\gamma(n-1)) \\
    =&\frac{i\omega}{2} R_{u_2}(n)R_{u_1}(n)\sin (\eta_{u_1}(n)-\eta_{u_2}(n))
 \end{align*}
 The Wronskian is nonzero and independent of $n$, but we know from the above that the quantities
 \[
 \frac{R_{u_1}(n)}{R_{u_2}(n)}, \qquad \sin(\eta_{u_1}(n) - \eta_{u_2}(n)), 
 \]
have nonzero limits as $n \to \infty$. From our final formula for the Wronskian, it then follows that $R_{u_1}^2(n)$, and then $R_{u_1}(n)$ has a nonzero limit as $n\to \infty$.
\end{proof}
\end{section}

\begin{section}{Pr\"ufer variables for the Szeg\H o recursion}\label{PruferSzego}
We begin by noting that $C^2 = I$ and
\[
Cz^{-1/2}A(\alpha,z)C=z^{-1/2}A(\alpha,z)
\]
so, since $v(n)$ is a solution for (\ref{OPUC.tmatrix}), so is $v^*(n)$. Analogously, denote $u^*(n) = Cu(n)$. The condition \eqref{u0} ensures that $u^*(0)=Cu(0)=u(0)$, so $u^*(n) = u(n)$ for all $n$.

Let us assume that $v(0), v^*(0)$ are linearly independent. Then $v(n), v^*(n)$ are linearly independent for any $n$. Thus, there exist complex numbers $Z(n),\mathfrak s(n)$ such that
\begin{equation}
u(n)=Z(n) v(n)+\mathfrak s(n) v^*(n).
\end{equation}
Applying $C$ to both sides of this equation we get
\[
u^*(n)=\overline{Z(n)} v^{*}(n)+\overline{\mathfrak s(n)} v(n).
\]
This implies that $\mathfrak s(n)=\overline{Z(n)}$ and proves \eqref{uvequation}.

Given $f,g$ two sequences in $\mathbb C^2$, we define their Wronskian as 
\begin{equation}
W(f,g)(n)=  f_2(n)g_1(n)- g_2(n)f_1(n).
\end{equation}
Our first order of business is to verify constancy of the Wronskian.

\begin{prop}
For $f, g$ two solutions of the Szeg\H o recursion with the same sequence $\alpha_n$, and $n$ a positive integer, $W(f,g)(n)=W(f,g)(n-1)$.
\end{prop}

\begin{proof}
Let us write
\[M(n)=
\begin{pmatrix}
g_1(n)&f_1(n)\\
g_2(n)& f_2(n)
\end{pmatrix}
\]
and note that $W(f,g)(n) = \det M(n)$. Then $M(n) = z^{-1/2} A(\alpha_{n-1},z) M(n-1)$ implies that $\det M(n) = \det M(n-1)$.
\end{proof}

Let us write $\omega$ as the Wronskian of $v$ and $v^*$. Note that $\omega$ will be a nonzero real constant, due to the assumption that $v, v^*$ are linearly independent. We can thus write
\begin{equation}\label{OPUC.Wronskian}
\omega=W(v,v^*)= v_2(n)(v^*)_1(n)- (v^*)_2(n)v_1(n)=\vert v_2(n)\vert^2-\vert v_1(n)\vert^2.
\end{equation}
Since $v$, $v^*$ are solutions corresponding to the same sequence $\alpha_n$, this expression is $n$-independent.

From (\ref{uvequation}) we can write
\[
Z(n)=\frac{W(u,v^*)(n)}{\omega }
\]

Since $u$ and $v^*$ are solutions of Szeg\H o recursions corresponding to different sequences of Verblunsky coefficients, their Wronskian will not be $n$-independent. We denote
\begin{equation}
\rho'_n=\sqrt{1-\vert \alpha_n+\alpha'_n\vert^2}-\rho_n.
\end{equation}
and compute
\begin{align}
&W(u,v^*)(n+1)-W(u,v^*)(n)\nonumber\\
=& u_2(n+1)(v^*)_1(n+1)- (v^*)_2(n+1)u_1(n+1) - u_2(n)(v^*)_1(n) + (v^*)_2(n)u_1(n)\nonumber\\
=& \frac 1{\rho_n(\rho_n+\rho'_n)} \Bigl[ 
 (\bar\alpha_n \alpha_n' + \rho_n \rho_n') u_1(n) \overline{v_1(n)} - z \alpha_n' u_1(n) \overline{v_2(n)} \nonumber \\
& \qquad\qquad\qquad + \bar z \bar\alpha'_n u_2(n) \overline{v_1(n)} - (\alpha_n \bar\alpha_n' + \rho_n \rho'_n) u_2(n) \overline{v_2(n)} \Bigr]
\label{WronskianDifferenceOPUC}
\end{align}
The last line is obtained by using Szeg\H o recursion to substitute $u_1(n+1), u_2(n+1), (v^*)_1(n+1), (v^*)_2(n+1)$ and simplifying the resulting expression.

Dividing by $\omega Z(n)$ gives
\begin{align*}
\frac{Z(n+1) }{Z(n)} -1 = &  \frac 1{\omega\rho_n(\rho_n+\rho'_n) Z(n)} \Bigl[ 
 (\bar\alpha_n \alpha_n' + \rho_n \rho_n') u_1(n) \overline{v_1(n)} - z \alpha_n' u_1(n) \overline{v_2(n)} \nonumber \\
& \qquad\qquad\qquad\qquad + \bar z \bar\alpha'_n u_2(n) \overline{v_1(n)} - (\alpha_n \bar\alpha_n' + \rho_n \rho'_n) u_2(n) \overline{v_2(n)} \Bigr]
\end{align*}
From (\ref{uvequation}),
\begin{align}
\frac{u_1(n)}{Z(n)}=& \frac {Z(n)v_1(n)+\overline{Z(n)} (v^*)_1(n)}{Z(n)}  = v_1(n) + e^{-2i\eta(n)} \overline{v_2(n)} \label{OPUCu1}
\end{align}
and, similarly,
\begin{align}
\frac{u_2(n)}{Z(n)} =&\frac {Z(n)v_2(n)+\overline{Z(n)} (v^*)_2(n)}{Z(n)} = v_2(n) + e^{-2i\eta(n)} \overline{v_1(n)}. \label{OPUCu2}
\end{align}
Plugging these into the previous formula proves the following theorem.

\begin{thm}\label{OPUCPrufer}
\begin{align}
\frac{Z(n+1) }{Z(n)} -1 = &  \frac 1{\omega\rho_n(\rho_n+\rho'_n)} \Bigl[ 
 (\bar\alpha_n \alpha_n' + \rho_n \rho_n') \lvert v_1(n)\rvert^2  - z \alpha_n' v_1(n) \overline{v_2(n)} \nonumber \\
& \qquad\qquad\qquad\qquad + \bar z \bar\alpha'_n v_2(n) \overline{v_1(n)} - (\alpha_n \bar\alpha_n' + \rho_n \rho'_n) \lvert v_2(n) \rvert^2 \nonumber \\
& \qquad\qquad\qquad\qquad + e^{-2i\eta(n)} \Bigl(   (\alpha_n \bar\alpha_n'  - \bar\alpha_n \alpha_n') \overline{v_2(n)}\overline{v_1(n)} \nonumber \\
& \qquad\qquad\qquad\qquad + z \alpha_n' \overline{v_2(n)}^2 - \bar z \bar\alpha'_n \overline{v_1(n)}^2  \Bigr)\Bigr] \label{recursionCMV}
\end{align}

\begin{remark}
For $\alpha_n \equiv 0$, $u_0 = \begin{pmatrix} 1 \\ 1 \end{pmatrix}$, $v_n = \begin{pmatrix} z^{n/2} \\ 0 \end{pmatrix}$, this reduces to the Pr\"ufer variables in \cite[Section 10.12]{SimonOPUC2}.
\end{remark}
\end{thm}

\begin{lemma}\label{InitialConditionsTrickCMV}
Fix $z$ and a solution $v$ of \eqref{OPUC.tmatrix}. For different solutions $u$ of\eqref{OPUC.perturbed}, we denote the corresponding Pr\"ufer variables by $Z_u$, $R_u$, $\eta_u$. If there exists a sequence $\mathcal A(n)$ independent of $u$ such that the series
\begin{equation}\label{ml01CMV}
\sum_{n=1}^\infty \left( \log \frac{Z_u(n+1)}{Z_u(n)}  - \mathcal A(n) \right)
\end{equation}
converges uniformly in $u$, then $R_u(n)$ converges as $n\to\infty$ for any $u$ and there is no subordinate solution at $z$.
\end{lemma}
  
\begin{proof}
For two solutions $u_1, u_2$, we consider the Wronskian of $u_1$ and $u_2$. Using linearity of the Wronskian and \eqref{uvequation} for $u_1, u_2$, we get
\begin{align*}
W(u_1,u_2)(n) = & Z_{u_1}(n) Z_{u_2}(n) W(v,v)(n) + Z_{u_1}(n) \overline{Z_{u_2}(n)} W(v,v^*)(n) \\
&  + \overline{Z_{u_1}(n)} Z_{u_2}(n) W(v^*,v)(n) + \overline{Z_{u_1}(n) Z_{u_2}(n)} W(v^*,v^*)(n) \\
= &  (Z_{u_1}(n) \overline{Z_{u_2}(n)} - \overline{Z_{u_1}(n)} Z_{u_2}(n)) \omega \\
=&2i \omega R_{u_1}(n)R_{u_2}(n) \sin(\eta_{u_1}(n)-\eta_{u_2}(n))
\end{align*}
The Wronskian is nonzero and independent of $n$.
 
 As in the proof of Lemma~\ref{InitialConditionsTrick}, we conclude that for any solution $u_1$ we can find a solution $u_2$ such that the sequences
\[
\frac{R_{u_1}(n)}{R_{u_2}(n)}, \qquad  \sin(\eta_{u_1}(n) - \eta_{u_2}(n))
\]
converge to nonzero limits as $n\to \infty$. From our final formula for the Wronskian, it then follows that $R_{u_1}^2(n)$, and then $R_{u_1}(n)$ has a nonzero limit as $n\to \infty$.

Since $R_{u_1}(n) / R_{u_2}(n)$ has a nonzero limit, there are no subordinate solutions.
\end{proof}

\begin{remark}\label{CMVinitcond}
While we consider solutions $u$ with initial conditions of the form \eqref{u0} which are convenient for our analysis, standard references \cite{SimonOPUC1,SimonOPUC2} single out solutions $\tilde u$ with initial conditions of the form
\[
\tilde u(0) = \begin{pmatrix} 1 \\ \lambda \end{pmatrix}, \qquad \lambda \in \partial \mathbb{D}.
\]
By linearity of the Szeg\H o recursion, with $\kappa = \sqrt \lambda$, we have $\tilde u(0) = \bar\kappa u(0)$ so $\tilde u(n) = \bar\kappa u(n)$. Thus, boundedness of solutions, subordinacy, and similar conclusions about the asymptotics carry over immediately from one family of solutions to the other.
\end{remark}

As the final topic of this section, we explain how the current setup pertains to eigensolutions of the CMV matrix.

\begin{remark}\label{GZremark}
As proved in \cite{GZ}, for any eigensolution $U$ of the CMV matrix $\mathcal{C}$, there is an eigensolution $V$ of $\mathcal{C}^T$,
\[
\mathcal{C} U = z U, \quad \mathcal{C}^T V = z V,
\]
such that
\[
\begin{pmatrix} U_n \\ V_n \end{pmatrix}
= R_n(\alpha_n,z) \begin{pmatrix} U_{n-1} \\ V_{n-1} \end{pmatrix}
\]
for all $n$, where
\[
R_n(\alpha_n,z) =\begin{cases}
\frac{1}{\rho_n}
\begin{pmatrix}
-\overline{\alpha_n} & z \\
z^{-1} & - \alpha_n
\end{pmatrix}, & n\text{ odd}, \\
\frac{1}{\rho_n}
\begin{pmatrix}
-\alpha_n & 1 \\
1 & - \overline{\alpha_n}
\end{pmatrix}, & n\text{ even}.
\end{cases}
\]
It was observed in \cite{DFLY} that
\begin{equation}\label{DFLYrelation}
z^{-1} A(\alpha_{2k+1},z) A(\alpha_{2k},z)
=
D(z) R_{2k+1}(\alpha_{2k+1},z) R_{2k}(\alpha_{2k},z) D\left(z\right)^{-1},
\end{equation}
where
\[
D(z) = \begin{pmatrix}
1 & 0 \\
0 & z
\end{pmatrix}.
\]
\eqref{DFLYrelation} relates Szeg\H o recursion with eigensolutions of $\mathcal{C}$. For example, if at some $z$, all solutions $u$ of Szeg\H o recursion are bounded, then iterating \eqref{DFLYrelation} implies that the sequences $(U_n, V_n)^T$ are bounded, so the eigensolutions $U_n$ are bounded.
\end{remark}
\end{section}

\begin{section}{Random decaying perturbations of Jacobi and CMV matrices} \label{SrandomJacobi}
\label{SectionRandom}

\begin{proof}[Proof of Theorem \ref{TrandomJacobi}]
For any $E \in (u,v)$, let us pick the solution $\varphi$ of \eqref{varphi} with $\varphi(0)=1$, $a_1\varphi(1)=i$. This is a somewhat arbitrary choice, but it ensures that $\omega = 2$ is bounded away from $0$. Together with
\[
\begin{pmatrix} 
a_{n+1} \varphi(n+1) \\ a_n \varphi(n) 
\end{pmatrix}
= T_n^E \begin{pmatrix} 
a_1 \varphi(1) \\ \varphi(0) 
\end{pmatrix},
\]
this implies that $\varphi$ is bounded for $E \in (u,v)$ and, moreover, if we define
\begin{equation}\label{PhiE}
\Phi(E) = \frac{\lVert \varphi\rVert_\infty^2}{\lvert\omega\rvert},
\end{equation}
then for any $\epsilon >0$,
\begin{equation}\label{PhiEbdd}
\sup_{E \in (u+\epsilon, v- \epsilon)} \Phi(E) < \infty.
\end{equation}
Denoting by $\tilde T_n^E$ transfer matrices for the perturbed potential, we wish to use the following criterion of Last--Simon \cite{LastSimon99}: if 
\begin{equation}\label{LastSimonCriterion}
\liminf_{n\to\infty}  \int_{u+\epsilon}^{v-\epsilon} \lVert \tilde T_n^E \rVert^4 dE < \infty,
\end{equation}
then $\tilde{\mathcal{J}}$ has purely a.c.\ spectrum on $(u+\epsilon, v-\epsilon)$ and $(u+\epsilon, v-\epsilon)$ is in the essential support of the a.c.\ spectrum. In \cite{LastSimon99} this is proved for discrete Schr\"odinger operators, but as remarked elsewhere (e.g. in \cite{KaluzhnyLast07}), the statement extends to the current setting essentially by the same proof.

Let us fix an initial condition $Z(1)$ independent of $E$. We will first show that, almost surely,
\begin{equation}\label{R4intfinite}
\liminf_{n\to\infty} \int_{u+\epsilon}^{v-\epsilon} R_n(E)^4 dE < \infty. 
\end{equation}
This will be done by applying the simple argument of \cite[Theorem 8.1]{KiselevLastSimon98} to our generalized Pr\"ufer variables.

From Theorem \ref{t.OPRL}, we have
\begin{align} R(n+1)^4  = & {R(n)}^4 \biggl(  1- 4 \Re \Bigl( \frac{i}{\omega} b_{n+1}'\vert \varphi(n)\vert^2 e^{-2i(\eta(n)+\gamma(n))}  \nonumber \\
&+\frac{i}{\omega}  a_{n}'\vert \varphi(n-1)\vert\cdot \vert \varphi(n)\vert e^{i(\gamma(n-1)-\gamma(n))}\nonumber \\
&-\frac{i}{\omega}  a_{n}'\vert \varphi(n-1)\vert\cdot \vert \varphi(n)\vert e^{-2i\eta(n)}e^{-i(\gamma(n-1)+\gamma(n))}\nonumber \\
&+\frac{i}{\omega} a_{n}' (1-e^{-2i(\eta(n)+\gamma(n))})\vert \varphi(n-1)\vert\cdot \vert \varphi(n)\vert e^{-i(\gamma(n-1)-\gamma(n))} \Bigr)\nonumber  \\
& + O\left( \Phi(E) ( \lvert a_n' \rvert^2 + \lvert b_{n+1}'\rvert^2)\right) \biggr). \label{R4}
\end{align}
Since random variables $R(n), \eta(n)$ depend only on $a'_1, \dots, a'_{n-1}, b'_1, \dots, b'_n$, they are independent from the random variables $a'_n, b'_{n+1}$, so the first order terms in $a_n'$ or $b'_{n+1}$ have zero expectation; for example, with $X(n) = R(n)^4  \frac{i}{\omega} \vert \varphi(n)\vert^2 e^{-2i(\eta(n)+\gamma(n))}$,
\[
\mathbb{E} \left( \Re \left( b'_{n+1} X(n) \right) \right) = \mathbb{E}(b'_{n+1}) \mathbb{E} \left( \Re \left( X(n) \right) \right)  = 0.
\]
Thus, \eqref{R4} implies
\[
\mathbb{E} \left( R(n+1)^4 \right) \le \left(1 + C \Phi(E) \mathbb{E} ( \lvert a_n' \rvert^2 + \lvert b_{n+1}'\rvert^2) \right) \mathbb{E} \left( R(n)^4 \right).
\]
Iterating this inequality and using \eqref{L2randomcondition} and \eqref{PhiEbdd} implies
\[
\sup_{n\in\mathbb{N}} \mathbb{E} \left( \int_{u+\epsilon}^{v-\epsilon}  R(n)^4 dE \right) < \infty.
\]
Therefore, by Fatou's lemma, \eqref{R4intfinite} holds almost surely.

We go from here to \eqref{LastSimonCriterion} in the standard way: we pick two linearly independent solutions $u_1, u_2$ corresponding to two initial conditions,
\[
\begin{pmatrix}
(a_1 + a_1') u_1(1) & (a_1 + a_1') u_2(1) \\
u_1(0) & u_2(0) 
\end{pmatrix}
=
\begin{pmatrix}
1 & 0 \\
0 & 1
\end{pmatrix}.
\]
The corresponding Pr\"ufer variables have initial conditions $Z_1(1) =1$ and $Z_2(1) = i$. Using \eqref{rho2}, \eqref{R4intfinite} implies that almost surely,
\[
\liminf_{n\to\infty}  \int_{u+\epsilon}^{v-\epsilon} \left\lVert \begin{pmatrix} (a_{n+1} + a'_{n+1}) u_j(n+1) \\ u_j(n) \end{pmatrix} \right\rVert^4 dE < \infty,
\]
for $j=1,2$. Now
\[
\lVert T_n^E  \rVert^2 \le \left\lVert T_n^E \begin{pmatrix} 1 \\ 0 \end{pmatrix} \right\rVert^2  + \left\lVert T_n^E \begin{pmatrix} 0 \\ 1 \end{pmatrix} \right\rVert^2
\]
implies \eqref{LastSimonCriterion}. As explained above, this implies that almost surely, on $(u+\epsilon, v-\epsilon)$, the spectral measure is mutually absolutely continuous with Lebesgue measure. Applying this conclusion to a countable sequence of $\epsilon \to 0$, say $\epsilon_n = n^{-1}$, concludes the proof.
\end{proof}

\begin{proof}[Sketch of proof for Theorem \ref{TrandomCMV}]
The proof is almost identical to that of the previous theorem. The analogue of \eqref{LastSimonCriterion} for CMV operators can be found as Theorem 10.7.5 of \cite{SimonOPUC2} (taking into account Remark 1 immediately after the theorem statement).
Starting from \eqref{recursionCMV}, and using algebraic manipulations such as
\begin{align}
\frac {\rho_n + \rho'_n}{\rho_n} & = \left( 1 - \frac{\bar \alpha_n \alpha'_n + \alpha_n \bar\alpha'_n + \alpha'_n \bar \alpha'_n}{\rho_n^2} \right)^{-1/2} \nonumber \\
& = \sum_{k=0}^\infty \binom{-1/2}{k} (-1)^k \frac{\left(\bar \alpha_n \alpha'_n + \alpha_n \bar\alpha'_n + \alpha'_n \bar \alpha'_n\right)^k}{\rho_n^{2k}}, \label{rho'toalpha'}
\end{align}
we can obtain an analogue of \eqref{R4} by considering terms that are $O(\alpha_n'^2)$ or linear in $\alpha_n'$.
\end{proof}
\end{section}

\begin{section}{Decaying oscillatory perturbations of periodic Jacobi and CMV matrices}\label{SectionDecaying}

In this section, we prove Theorems~\ref{maintheorem} and \ref{maintheoremCMV}.

Let $E$ lie in the interior of a band of the spectrum of the periodic Jacobi matrix $J$, and let $\varphi$ be a Floquet solution at $E$. Then we can write $\gamma(n)=\varpi(n)+kn$ where $k$ is the quasimomentum and $\vert \varphi(n)\vert$ and $\varpi(n)$ are $q$-periodic. Note that $\varphi$ is complex because $\pm k\in (0,\pi/q)$. Also note that the function $\Phi(E)$ defined by \eqref{PhiE} is a continuous function of $E$ on the interior of any band, so if we work on a compact interval $I$ in the interior of a band, then
\begin{equation}\label{Phiupperbound}
\tilde \Phi(I) = \sup_{E\in I} \Phi(E) < \infty.
\end{equation}

Let us make another preliminary remark. Since $a_n', b_n'$, we may take the average between the representation \eqref{1.13} and its complex conjugate with no change to the assumptions of the theorem. Thus, we may assume that for every term in \eqref{1.13}, there is another term which is precisely its complex conjugate.

We now rewrite the result of Theorem~\ref{t.OPRL} in a convenient form.
\begin{align} &\frac{Z(n+1)}{Z(n)}\nonumber\\
=&1-\frac{i}{\omega} \frac{a_n}{a_n+a_n'} b_{n+1}'\vert \varphi(n)\vert^2 (e^{-2i\eta(n)}e^{-2i\varpi(n)-2ikn} -1)\nonumber\\
&+\frac{i}{\omega}  a_{n}'\vert \varphi(n-1)\vert\cdot \vert \varphi(n)\vert e^{i(\varpi(n-1)-\varpi(n)-k)}\nonumber\\
&-\frac{i}{\omega}  a_{n}'\vert \varphi(n-1)\vert\cdot \vert \varphi(n)\vert e^{-2i\eta(n)} e^{-i(\varpi(n)+\varpi(n-1)+(2k-1)n)}\nonumber\\
&+\frac{i}{\omega}\frac{a_n}{a_n+a_n'} a_{n}' (1-e^{-2i\eta(n)}e^{-2i\varpi(n)-2ikn})\vert \varphi(n-1)\vert\cdot \vert \varphi(n)\vert e^{i(\varpi(n)-\varpi(n-1)+k)} \label{r(n+1)/r(n)}
\end{align}

Note that coefficient stripping does not affect the conclusions of our theorem. Since we are working with a decaying perturbation, this means we can assume that for all $n$,
\[
\left\lvert \frac{a_n'}{a_n} \right\rvert \le \frac 12\quad \text{and} \quad \lvert a_n' \rvert , \lvert b_n' \rvert \le \frac 1{20} \tilde\Phi(I).
\]
(these conditions are trivially true for $n\ge n_0$; by coefficient stripping $n_0$ times, they become true for all $n$). These assumptions ensure that some of the Taylor expansions below are justified;  for instance, by \eqref{Phiupperbound}, they ensure that
\[
\left\lvert \frac{Z(n+1)}{Z(n)} - 1 \right\rvert \le \frac12
\]
so we can take the $\log$ of \eqref{r(n+1)/r(n)}, using the usual branch of $\log$ on $\{ z\in \mathbb{C} \mid \lvert z-1\rvert < 1\}$ with $\log 1=0$. Using the Taylor series of $\log$ to expand the right-hand side, and using the geometric series
\[
\frac{a_n}{a_n+a_n'}=\frac{1}{1+a_n'/a_n}=1-\frac{a'_n}{a_n}+\left(\frac{a'_n}{a_n}\right)^2-\ldots
\]
we can write
\begin{equation}\label{logZ}
\log \frac{Z(n+1)}{Z(n)} = P(n)  + Q(n),
\end{equation}
where $P(n)$ collects all terms with at most $p-1$ factors of $a_n'$ and $b'_{n+1}$,
\begin{equation}\label{ml50}
P(n) = \sum_{\substack{K,L\geq 0\\ 1 \le K+L \le p-1}}  \sum_{M=0}^{1} \zeta_{K,L,M}(n) {a_n'}^K {b_{n+1}'}^L  e^{-2iM(kn+ \eta(n))}
\end{equation}
and $\zeta_{K,L,M}(n)$ are $q$-periodic sequences which don't depend on $a'$ or $b'$. The remainder $Q(n)$ collects all terms with $p$ or more factors of $a'_n$ and $b'_{n+1}$, so  $Q(n) \in \ell^1$; $Q(n)$ will be merely an inconsequential remainder in what follows.

Using \eqref{1.13}, we expand \eqref{ml50} into a sum of terms of the form
\[
\xi_{M,l_1, \dots, l_{K+L}} (n) c_1 \cdots c_{K+L} e^{-i(\phi_{l_1} + \dots + \phi_{l_{K+L}})} \varsigma_n^{(l_1)} \dots \varsigma_n^{(l_{K+L})} e^{-2iM (k n + \eta(n))}
\]
where each $\xi_{M,l_1, \dots, l_{K+L}}$ is equal to $\zeta_{K,L,M}$ or to $0$. To abbreviate the expressions, for
\[
\lj = (l_1, \dots, l_{J}) \in \mathbb{N}^J,
\]
we denote
\begin{align*}
\varsigma_n^{(\lj)} & = \varsigma_n^{(l_1)} \dots \varsigma_n^{(l_{J})} \\
\phi_{\lj} & = \phi_{l_1} + \dots + \phi_{l_{J}} \\
c_{\lj} & = c_{l_1} \cdots  c_{l_{J}}
\end{align*}
and we can write \eqref{logZ} as
\begin{equation}\label{P}
\log \frac{Z(n+1)}{Z(n)} = \sum_{J = 1}^{p-1}\sum_{M=0}^1  \sum_{\lj \in \mathbb{N}^J} \xi_{M,\lj}(n) c_{\lj} e^{-i \phi_\lj n} \varsigma_n^{(\lj)} e^{-2i M (k n+ \eta(n))} + Q(n). 
\end{equation}
In the same fashion, denoting $Z(n+1) / Z(n) = 1+ w$, starting from (by \eqref{ZReta})
\[
e^{2iM(\eta(n) - \eta(n+1))} = \left( \frac {1+\bar w}{1+w} \right)^M,
\]
expanding $(1+w)^{-M}$ in powers of $w$, then expanding $w$ by using \eqref{1.13}, and collecting into $Q_M \in \ell^1$ all the products with at least $p$ factors, we obtain
\begin{equation}\label{eta2}
e^{2iM(\eta(n)-\eta(n+1))}= 1 - \sum_{J = 1}^{p-1}\sum_{m=-M}^J  \sum_{\lj \in \mathbb{N}^J} \omega_{M,m,\mathbf j}(n) c_\lj e^{-i \phi_\lj n} \varsigma_n^{(\lj)} e^{-2i m (k n+ \eta(n))} - Q_M(n).
\end{equation}

Note that the sum in $m$ goes only from $-M$, since positive powers of $e^{2i(kn+\eta(n))}$ can only come from $(1+\bar w)^M$, at most $M$ of them; note also that it doesn't go beyond $J$, since every factor of $e^{-2i(kn+\eta(n))}$ is accompanied by at least one $a_n'$ or $b_{n+1}'$.

It is obvious that the $\xi_{M,\lj}$ and $\omega_{M,m,\lj}$ are bounded, i.e.\ that for any fixed $M, m, J$, there is a constant $C$ depending only on $M, m, J$ such that
\begin{equation}\label{xiomegabounds}
\sup_{\lj \in \mathbb{N}^J} \lVert \xi_{M,\lj} \rVert_\infty \le C \Phi(E)^J, \quad \sup_{\lj \in \mathbb{N}^J} \lVert \omega_{M,m,\lj} \rVert_\infty \le C \Phi(E)^J
\end{equation}
since there are, in fact, only finitely many distinct sequences among the $\xi_{M,\lj}$ and $\omega_{M,m,\lj}$, and they are all $q$-periodic, and contain at most $J$ factors of $\lvert \varphi(\cdot) \varphi(\cdot) \rvert / \omega$.

\eqref{P} and \eqref{eta2} will be the key formulas in what follows. They both involve seemingly complicated sums, but note that both sums are (infinite) linear combinations of terms which are all of the same form: every term is a periodic factor multiplied by an oscillation and a sequence of bounded variation. We control such sequences using a discrete integration by parts to integrate the skew-periodic part and differentiate the bounded variation sequence.

To integrate the skew-periodic part, we use the following.

  \begin{lemma}\label{fglemma}
  Given a $q$-periodic sequence $f(n)$ and a real number $\kappa$ such that $\kappa q\notin 2\pi \mathbb{Z}$, there exists a unique $q$-periodic sequence $g(n)$ such that 
  \begin{equation}\label{fg.equation}
  e^{i\kappa n}g(n)-e^{i\kappa (n-1)}g(n-1)=e^{i\kappa n}f(n).
  \end{equation}
  We will denote  $\Lambda_\kappa (f):=g$.
Moreover,
\[
\lVert \Lambda_\kappa (f) \rVert_\infty \le \frac{ q \lVert f \rVert_\infty}{\lvert e^{i \kappa q} - 1 \rvert}.
\] 
  \end{lemma}

  \begin{proof}
The sequence $g(n)$ is uniquely determined by \eqref{fg.equation} and by the value of $g(0)$ (and of course, $f$). Moreover, periodicity dictates that
\begin{equation}\label{g0gq}
g(0) = g(q) = e^{-i\kappa q} \left( g(0) + \sum_{n=1}^{q} e^{i\kappa n} f(n) \right)
\end{equation}
which gives the only possibility for $g(0)$,
\begin{equation}\label{g0}
g(0) = \frac{\sum_{n=1}^{q} e^{i\kappa n} f(n)}{e^{i\kappa q}-1}.
\end{equation}
Conversely, \eqref{g0} implies \eqref{g0gq}, and the sequence $g$ given by \eqref{fg.equation} and \eqref{g0gq} is obviously $q$-periodic if $f$ is.

For $m \in \{0,1,\dots, q-1\}$,
\begin{align*}
e^{i\kappa m} g(m) & = g(0) + \sum_{n=1}^{m} e^{i\kappa n} f(n)  = \frac 1{e^{i\kappa q}-1} \left( e^{i\kappa q} \sum_{n=1}^m e^{i\kappa n} f(n) + \sum_{n=m+1}^q e^{i\kappa n} f(n) \right)
\end{align*}
from which the estimate on $\lVert g \rVert_\infty$ is immediate.
\end{proof}
  
The following lemma is tailor-made to control terms such as those that appear in \eqref{P}.
    
\begin{lemma}\label{mainlemma}
Let $M\in \Z$ and $\phi\in \mathbb{R}$. Let $f(n)$ be $q$-periodic and let $\kappa = -2Mk-\phi \notin 2\pi\mathbb{Z}$. If the sequence $\varsigma$ has bounded variation and $\varsigma_n\to 0$, then
\begin{align}
&\left \vert\sum_{n=1}^N \left( f(n) e^{-i\phi n}  e^{-2iM(kn + \eta(n))} \varsigma_{n}
- (\Lambda_\kappa f)(n) e^{-i\phi n}  e^{-2iM(kn + \eta(n))}  \varsigma_{n}  (1 - e^{2iM(\eta(n) - \eta(n+1))}) \right) \right\vert\nonumber\\
\leq & 4 \lVert \Lambda_\kappa f \rVert_\infty \Var(\varsigma)\label{mainlemma.equation}
\end{align}
where $\Var(\varsigma)$ stands for the variation of the sequence $\varsigma$.
\end{lemma}

\begin{proof}
Denoting $g = \Lambda_\kappa f$ and using \eqref{fg.equation}, we can rewrite the sum in the left-hand side of \eqref{mainlemma.equation} as
\begin{align*}
& \sum_{n=1}^N \left(  ( g(n)e^{i\kappa n} - g(n-1)e^{i\kappa (n-1)} ) e^{-2iM\eta(n)} \varsigma_{n}
- g(n) e^{i\kappa n} \varsigma_{n}  (e^{-2iM\eta(n)}  - e^{ - 2 i M\eta(n+1)}) \right) \\
& = \sum_{n=1}^N \left( - g(n-1)e^{i\kappa (n-1)} e^{-2iM\eta(n)} \varsigma_{n}
+ g(n) e^{i\kappa n} \varsigma_{n}  e^{ - 2 i M\eta(n+1)} \right).
\end{align*}
This is the sum of two sums; the first is the telescoping sum,
\begin{align*}
& \left\lvert \sum_{n=1}^N \left( - g(n-1)e^{i\kappa (n-1)} e^{-2iM\eta(n)} \varsigma_{n}
+ g(n) e^{i\kappa n} \varsigma_{n+1}  e^{ - 2 i M\eta(n+1)} \right) \right\rvert \\
&  = \left\lvert g(N) e^{i\kappa N} \varsigma_{N+1}  e^{ - 2 i M\eta(N+1)} - g(0) \varsigma_{1}  e^{ - 2 i M\eta(1)} \right\rvert \\
& \le 2 \lVert g\rVert_\infty \lVert \varsigma\rVert_\infty
\end{align*}
and the second is a sum bounded by bounded variation,
\[
\left\lvert \sum_{n=1}^N 
g(n) e^{i\kappa n} (\varsigma_n - \varsigma_{n+1} ) e^{ - 2 i M\eta(n+1)} \right\rvert \le 2 \lVert g \rVert_\infty \sum_{n=1}^N \lvert \varsigma_n - \varsigma_{n+1} \rvert \le 2 \lVert g \rVert_\infty \Var(\varsigma).
\]
Since $\varsigma$ is decaying, $\lVert \varsigma\rVert_\infty \le \Var(\varsigma)$, so \eqref{mainlemma.equation} follows from the previous estimates.
\end{proof}

We now have the tools necessary to start an iterative procedure. To any term of the form
\begin{equation}\label{sampleterm}
f_{M,\lj}(n) c_{\lj} e^{-i \phi_\lj n} \varsigma_n^{(\lj)} e^{-2i M (k n+ \eta(n))} 
\end{equation}
we can apply the previous lemma, to replace it by
\[
(\Lambda_{-2Mk-\phi_\lj} f_{M,\lj})(n) c_{\lj} e^{-i \phi_\lj n} \varsigma_n^{(\lj)} e^{-2i M (k n+ \eta(n))}  (1 - e^{2iM(\eta(n) - \eta(n+1))})
\]
and then, using \eqref{eta2} to express $1 - e^{2iM(\eta(n) - \eta(n+1))}$, to get to another sum of terms of the form \eqref{sampleterm}, but with longer vectors $\lj$. This leads to a recursion relation for $f_{M,\lj}$. Denote
\begin{equation}\label{g.recursion}
g_{M,\lj} = \Lambda_{-2Mk-\phi_\lj} f_{M,\lj}.
\end{equation}
The recursion relation for $f$ is given by
\begin{equation}\label{f.recursion}
f_{M,l_1,\dots, l_J} = \xi_{M,l_1,\dots, l_J} + \sum_{j=1}^{J-1} \sum_{m=0}^{j} g_{m,l_1, \dots, l_{j}} \odot  \omega_{m,M-m,l_{j+1}, \dots, l_J},
\end{equation}
where $\odot$ stands for a product symmetrized over $l_1,\dots,l_J$,
\[
g_{m,l_1, \dots, l_{j}} \odot  \omega_{m,M-m,l_{j+1}, \dots, l_J} = \frac 1{J!} \sum_{\pi \in S_J} g_{m,l_{\pi(1)}, \dots, l_{\pi(j)} } \omega_{m,M-m,l_{\pi(j+1)}, \dots, l_{\pi(J)}}
\]

The proof of Theorem~\ref{maintheorem} is immediate from the following two lemmas.

\begin{lemma}\label{hausdorffdimlemma}
\begin{enumerate}[(a)]
\item Let $\alpha, \beta \in (0,1]$ and let $\nu$ be a finite U$\alpha$H measure on $\mathbb{R}$ such that $\supp \nu$ lies in the interior of a band. Let $M, J \in \mathbb{N}$ with $1\le M\le J$. If $\alpha > J\beta$, then
\begin{equation}\label{CMJ}
\int \lVert g_{M,l_1, \dots, l_J}(k) \rVert_\infty^\beta d\nu(k) \le C_{M,J,\beta,\nu}.
\end{equation}
If instead $\alpha >(J-1)\beta$, then
\begin{equation}\label{CMJ2}
\int \lVert f_{M,l_1, \dots, l_{J}}(k) \rVert_\infty^\beta d\nu(k) \le \tilde C_{M,J,\beta,\nu}.
\end{equation}

\item For any $M, J \in \mathbb{N}$ with $1\le M\le J$, the small divisor condition
\begin{equation}\label{smalldivisor}
\sum_{\lj \in \mathbb{N}^J} \lVert g_{M, \lj} \rVert_\infty \lvert c_\lj \rvert < \infty
\end{equation}
holds for all $k\notin S$, for some set $S$ with $\dim_H S \le J\beta$.
\end{enumerate}
\end{lemma}

\begin{lemma} \label{lemmafixedE} Assume the notation and assumptions of Therorem~\ref{maintheorem}. Assume further that for a given $E\in \R$ in the equation \eqref{u}, and for a sequence of frequencies $\{\phi_j\}_{j\in\mathbb N}$ , the following small divisor condition holds for any integers $m, j\in\mathbb{N}$ with $1\le m \le j \le p-1$: 
\[
\sum_{\lj \in \mathbb{N}^j}\rVert c_{\lj}g_{m,\lj} \rVert_\infty<\infty.
\]
Then solutions of \eqref{u} are bounded.
\end{lemma}

\begin{proof}[Proof of Lemma~\ref{hausdorffdimlemma}] 
(a) We first recall a basic fact. Since $\nu$ is a finite U$\alpha$H measure, for any $\phi \in \mathbb{R}$
\begin{equation}\label{CMJ3}
\int \frac 1{\lvert k - \phi \rvert^\beta} d\nu(k) \le D_\beta
\end{equation}
where $D_\beta$ is a finite constant independent of $\phi$; see, e.g.,\cite[Lemma 4.1]{Lukic-infinite}.

Note that, since $\supp \nu$ lies in the interior of a band,
\[
\sup_{k \in \supp \nu} \frac 1{\lvert e^{ikq} - 1\rvert} < \infty, \qquad  \sup_{k \in \supp \nu} \Phi(E(k)) < \infty.
\]
Now we can prove \eqref{CMJ} and \eqref{CMJ2} by induction. The induction is fueled by two inequalities which follow from \eqref{g.recursion} and \eqref{f.recursion},
\begin{equation}\label{g.inequality}
\lVert g_{M,\lj} \rVert_\infty \le \frac q{\lvert e^{ikq} - 1 \rvert} \lVert f_{M,\lj}\rVert_\infty
\end{equation}
and
\begin{equation}\label{f.inequality}
\lVert f_{M,l_1,\dots, l_J} \rVert_\infty^\alpha \le \lVert\xi_{M,l_1,\dots, l_J}\rVert_\infty^\alpha +  \sum_{j=1}^{J-1} \sum_{m=0}^{j}  \lVert g_{m,l_1, \dots, l_j} \rVert_\infty^\alpha \odot \lVert \omega_{m,M-m,l_{j+1}, \dots, l_J}\rVert_\infty^\alpha.
\end{equation}

Assume that \eqref{CMJ} and \eqref{CMJ2} hold for values smaller than $J$. Since the norms of $\xi$'s and $\omega$'s are uniformly bounded and $j\beta < \alpha$ for all $j\le J-1$, integrating \eqref{f.inequality} by $d\nu(k)$ we obtain \eqref{CMJ2}. Using H\"older's inequality, \eqref{g.inequality} implies
\[
\int \lVert g_{M,\lj} \rVert_\infty^\beta d\nu(k) \le \left( \int \left\lvert \frac q{e^{ikq-1}}\right\rvert^{J\beta} d\nu(k) \right)^{1/J} \left( \int \lVert f_{M,\lj} \rVert_\infty^{J\beta/(J-1)} d\nu(k) \right)^{(J-1)/J}.
\]
Both integrals on the right-hand side are bounded by \eqref{CMJ2} and \eqref{CMJ3}, which completes the inductive step.

(b) For any compact set $K$ that lies in the interior of a band, $\sup_K \Phi < \infty$. Thus, \eqref{smalldivisor} holds on $K$ everywhere except on a set of Hausdorff dimension at most $J\beta$, by the same argument as in the proof of Lemma 4.2 in \cite{Lukic-infinite}. Since the spectrum of $J$ can be covered (up to finitely many points) by countably many such compact sets $K$, the claim follows.
\end{proof}

Before we can prove Lemma \ref{lemmafixedE}, we need some preparatory work. We denote
 \begin{equation}
 \label{sigma}
 \sigma=  \sup_j \lVert \varsigma_j \rVert_p,
\end{equation} 
 which is finite by the assumptions of Theorem \ref{maintheorem}. We write
\begin{equation}\label{SJMn}
\mathcal S_{J,M}(n)=\sum_{\lj \in \mathbb{N}^J} f_{M,\lj}(n) c_{\lj} e^{-i \phi_\lj n} \varsigma_n^{(\lj)} e^{-2i M (k n+ \eta(n))} 
\end{equation}
where based on (\ref{ml50}) we see that $P(n)$ is a finite sum of terms that are a product of a $q$-periodic function and some $S_{ J,1}$. Thus it suffices for our purposes to show that $\sum_n S_{J,1}(n)$ converges for any value of $J$.

Our goal is to replace $S_{\tilde J,1}$ by replacing $J=\tilde J$ with higher values of $J$. Let us define 

\begin{equation}\label{Edefn}E_{J,M}=\sum_{l_1,\ldots,l_J=1}^\infty \vert c_{l_1}\ldots c_{l_J} g_{M,l_1,\dots,l_J}\vert,
\end{equation}
which is finite by the assumptions of Lemma \ref{lemmafixedE}, and let us note that
\[
\sum_{l_1,\ldots,l_J=1}^\infty \vert c_{l_1}\ldots c_{l_J} \vert  = \lVert c \rVert_1^J < \infty.
\]

\begin{lemma}

For any $0\leq M \leq J \leq p$, there is a finite constant $C_{J,M}$ such that
\begin{align}
\sum_{\lj \in \mathbb{N}^J}  \lVert f_{M,\lj}(n) c_{\lj} e^{-i \phi_\lj n} \varsigma_n^{(\lj)} e^{-2i M (k n+ \eta(n))} \rVert_\infty \le C_{J,M} \tau^J.
\label{J=2.bound}
\end{align}
In particular, the sum \eqref{SJMn} that defines $S_{J,M}(n)$ is absolutely convergent. If $J=p$, it also holds that
\begin{equation}\label{J=p.bound}
\sum_n \sum_{\lj \in \mathbb{N}^J}  \left\lvert f_{M,\lj}(n) c_{\lj} e^{-i \phi_\lj n} \varsigma_n^{(\lj)} e^{-2i M (k n+ \eta(n))}  \right\rvert \le C_{J,M} \sigma^p,
\end{equation}
so $S_{J,M}(n)$ is then also absolutely summable in $n$,
\[
\sum_n \left\lvert S_{J,M}(n) \right\rvert \le C_{J,M} \sigma^p.
\]
\end{lemma}
\begin{proof}
For $J\geq 2$, we use \eqref{f.inequality} to show
\begin{equation} \label{f.recursion.odot.abs}
\lVert f_{M,l_1,\dots, l_J} \rVert\leq \lVert\xi_{M,l_1,\dots, l_J}\rVert+
\sum_{j=1}^{J-1} \sum_{m=0}^{j}  \lVert g_{m,l_1, \dots, l_j} \rVert \odot \lVert \omega_{m,M-m,l_{j+1}, \dots, l_J}\rVert.
\end{equation}
Multiplying by $c_\lj$ and summing in $\lj \in \mathbb{N}^J$, we get
\begin{equation} \label{f.recursion.odot.abs2}
\sum_{l\in \mathbb{N}^J} \lVert f_{M,\lj} c_\lj \rVert \leq  C_{J,M},
\end{equation}
where, using \eqref{xiomegabounds} we get
\[C_{J,M} = \left( C \Phi(E)^J  \lVert c\rVert_1^J + \sum_{j=1}^{J-1} \sum_{m=0}^j  C \Phi(E)^{J-j} \lVert c \rVert_1^{J-j} E_{j,m} \right).
\]
Multiplying \eqref{f.recursion.odot.abs2} by $\lVert \varsigma_\lj \rVert_\infty \le \tau^J$, we obtain \eqref{J=2.bound}. If $J=p$, we note that by \eqref{sigma} and H\"older's inequality,
\[
\sum_n \lvert \varsigma_\lj(n) \rvert \le \sigma^p,
\]
so multiplying \eqref{f.recursion.odot.abs2} by this we conclude \eqref{J=p.bound}.
\end{proof}

Let us introduce, for integer $1\leq t\leq J-1$,
\begin{equation}\label{tildef.recursion}
\tilde f_{M,l_1,\dots, l_J}^{(t)} =  \sum_{m=0}^{t} g_{m,l_1, \dots, l_{t}} \odot  \omega_{m,M-m,l_{t+1}, \dots, l_J}.
\end{equation}

Using \eqref{f.recursion}, it is clear that 
\begin{equation}\label{fsum}
 f_{M,l_1,\dots, l_J}=\xi_{M,l_1,\dots, l_J} +\sum_{t=1}^{J-1}  \tilde f_{M,l_1,\dots, l_J}^{(t)}.
\end{equation}
Let us define also
\begin{equation}\label{tildeSJMn}
\tilde S_{J,M}^{(t)}=\sum_{\lj \in \mathbb{N}^J} \tilde f^{(t)}_{M,\lj}(n) c_{\lj} e^{-i \phi_\lj n} \varsigma_n^{(\lj)} e^{-2i M (k n+ \eta(n))}, 
\end{equation}
and additionally
\begin{equation}\label{Tdefn}
T_{J,M}(n)=\sum_{\lj \in \mathbb{N}^J} \xi_{M,l_1,\dots, l_J} c_{\lj} e^{-i \phi_\lj n} \varsigma_n^{(\lj)} e^{-2i M (k n+ \eta(n))}.
\end{equation}
It is clear from \eqref{SJMn} that 
\begin{equation}\label{Ssum}
 S_{J,M}=T_{J,M}+\sum_{t=1}^{J-1}\tilde S_{J,M}^{(t)}.
\end{equation}

\begin{lemma}\label{SJK.lemma}
For $t=1,\ldots p-1$, there is a constant $C$ so that
\[
\left\vert \sum_n \left( \sum_{M=1}^{t} S_{t,M}(n)-\sum_{J=t+1}^{p-1}\sum_{M=0}^{J} \tilde S^{(t)}_{J,M}(n)\right) \right\vert\leq C\sum_{M=1}^J  E_{J,M}\tau^t.
\]
\end{lemma}

\begin{proof}
We take \eqref{mainlemma.equation} with $f=f_{M,l_1,\ldots, l_t}$, multiply by $c_{l_1}\ldots c_{l_t}$ and then sum in $l_1\ldots l_t$. The lemma then follows by \eqref{g.recursion},\eqref{eta2}, \eqref{f.recursion}, \eqref{tildef.recursion} and \eqref{Edefn}.
\end{proof}

\begin{proof}[Proof of Lemma \ref{lemmafixedE}]

We start with the expression \eqref{P}. We can rewrite it as

\[S_{1,1}(n)+\sum_{J=2}^{p-1}\sum_{M=0}^1 T_{J,M}(n)\]

We apply Lemma \ref{SJK.lemma} to the $S_{1,1}$ term. Using \eqref{Ssum}, summing in $t=1,\ldots, p-1$ gets us 

\[S_{1,1}+\sum_{J=2}^{p-1}\sum_{M=0}^1 T_{J,M}\sim S_{2,0}+S_{2,1}+\sum_{J=3}^{p-1}\sum_{M=0}^1 T_{J,M}+\sum_{J=3}^{p-1}\sum_{M=0}^1 \tilde S^{(1)}_{J,M}.\]

By repeatedly applying Lemma \ref{SJK.lemma} to the $S$- terms, we eventually obtain

\[\left\vert \sum_n \left( S_{ 1,1}(n)+\sum_{J=2}^{p-1}\sum_{M=0}^1 T_{J,M}(n)-\sum_{M=1}^q S_{q,M}(n)-\sum_{j=2}^q S_{j,0}(n)\right)\right\vert\leq C\sum_{j=1}^{q-1}\sum_{m=1}^jE_{j,m}\tau^j.\]

We then apply (\ref{J=p.bound}) and use the triangle inequality to get

\[
\left \lvert\sum_{M=1}^q \sum_n S_{q,M}(n)\right\rvert
\leq \sum_{m=0}^{q-1}E_{q-1,m}\sum_{l=1}^\infty \vert c_l\vert\sigma^p+\sum_n \mathcal E\lVert \Phi(n)\rVert
\]
Therefore,
\[
\left\lvert \sum_n \left( \log \frac{Z(n+1)}{Z(n)} - \sum_{j=1}^p S_{j,0}(n) \right) \right\rvert \le \sum_{m=0}^{q-1}E_{q-1,m}\sum_{l=1}^\infty \vert c_l\vert\sigma^p+\sum_n \mathcal E\lVert \Phi(n)\rVert
\]
The right-hand side is finite; furthermore, the $\sum_{j=2}^q S_{j,0}(n)$ term is independent of $u$, so by  Lemma \ref{InitialConditionsTrick}, the sequence $R(n)$ is bounded, which concludes the proof.
\end{proof}

The proof of Theorem \ref{maintheoremCMV} is almost identical to the proof of Theorem \ref{maintheorem}.

\begin{proof}[Proof of Theorem \ref{maintheoremCMV}]
We consider \eqref{recursionCMV}, and using algebraic manipulations such as \eqref{rho'toalpha'}
we can write $\log\left(Z(n+1)/Z(n)\right)$ as a series in $\alpha'_n$, $\bar\alpha_n'$. We then wish to use \eqref{alpha'CMV} to write this in the form \eqref{P}; however, notice that $\bar\alpha_n'$ appears, so complex conjugates of $\varsigma_j$ would appear as well if we use \eqref{alpha'CMV}. To get exactly the form \eqref{P}, we therefore change the notation from \eqref{alpha'CMV} to
\[
\alpha'(n) =  \sum_{l=1}^\infty c_{2l-1} e^{-in\phi_{2l-1}} \varsigma^{(2l-1)}_{n}, \qquad \bar\alpha'(n) =  \sum_{l=1}^\infty c_{2l} e^{-in\phi_{2l}} \varsigma^{(2l)}_{n},
\]
and use this to obtain the form \eqref{P}; obviously, the new $c_l$ and $\varsigma^{(l)}$ obey the same conditions as before. We can then repeat the proof of Theorem \ref{maintheorem}, using Lemma \ref{InitialConditionsTrickCMV} in the final stage.

As in the proof of Theorem~\ref{maintheorem}, we thus conclude that there is a set $S$ with $\dim_H \le (p-1)\beta$ such that for $z \in \sigma_\ess(\mathcal{C})\setminus S$, there are no subordinate solutions of \eqref{OPUC.perturbed} and any solution $u$ of \eqref{OPUC.perturbed}, \eqref{u0} is bounded. This implies the conclusions of the theorem: boundedness of eigensolutions follows from Remark~\ref{GZremark}, and \eqref{acC} follows from Theorem 10.9.1 of \cite{SimonOPUC2}.
\end{proof}

\end{section}

\bibliographystyle{alpha}   
\bibliography{mybib}
\end{document}